\documentclass[manuscript,nonacm]{acmart}

 \usepackage{hyperref}
\usepackage{xcolor}
\usepackage{ifthen}
\usepackage{amsfonts}
\usepackage{amsmath}

\usepackage{graphicx}
\usepackage{amsthm}  

\newtheorem{theorem}{Theorem}
\newtheorem{lemma}{Lemma}

\usepackage{comment}
\usepackage{url}

\usepackage{amsmath} 

\title{Design and Detection of Covert  Man-in-the-Middle Cyberattacks  on Water Treatment Plants} 

\copyrightyear{2025}
\acmYear{2025}
\setcopyright{acmlicensed}\acmConference[RICSS '25]{Proceedings of the 2025 Workshop on Re-design Industrial Control Systems with Security}{October 13--17, 2025}{Taipei, Taiwan}
\acmBooktitle{Proceedings of the 2025 Workshop on Re-design Industrial Control Systems with Security (RICSS '25), October 13--17, 2025, Taipei, Taiwan}
\acmDOI{10.1145/3733823.3764516}
\acmISBN{979-8-4007-1911-0/2025/10}

\author{Victor Mattos}
\affiliation{%
  \institution{UFRJ}
  \city{Rio de Janeiro}
  \country{Brazil}
  \postcode{43017-6221}
}

\author{João Henrique Schmidt}
\affiliation{%
  \institution{UFRJ}
  \city{Rio de Janeiro}
  \country{Brazil}
  \postcode{43017-6221}
}

\author{Amit Bhaya}
\affiliation{%
  \institution{UFRJ}
  \city{Rio de Janeiro}
  \country{Brazil}
  \postcode{43017-6221}
}

\author{Alan Oliveira de Sá}
\affiliation{%
  \institution{LASIGE, Faculdade de Ciências, Universidade de Lisboa}
  \city{Lisbon}
  \country{Portugal}
  \postcode{-- }
}

\author{Daniel Sadoc Menasché}
\affiliation{%
  \institution{UFRJ}
  \city{Rio de Janeiro}
  \country{Brazil}
  \postcode{43017-6221}
}

\author{Gaurav Srivastava}
\affiliation{%
  \institution{Siemens Corporation,   755 College Rd E}
  \city{Princeton, NJ}
  \country{USA}
  \postcode{-- }
}

\thanks{This work was  partially supported  by CAPES, Finance Code 001; by FAPERJ under grant E-26/204.268/2024; and by the Conselho Nacional de Desenvolvimento Científico e Tecnológico (CNPq), Brazil, under grants 403601/2023-1,  315106/2023-9, and BPP/PQ-Sr  313335/2022-2, as well as  FCT through the LASIGE Research Unit, ref. UID/00408/2025.}

\begin{document}

\begin{abstract}
Cyberattacks targeting critical infrastructure --- such as water treatment facilities --- represent significant threats to public health, safety, and the environment. This paper introduces a systematic approach for modeling and assessing covert man-in-the-middle (MitM) attacks that leverage system identification techniques to inform the attack design. We focus on the attacker’s ability to deploy a covert controller, and we evaluate countermeasures based on the Process-Aware Stealthy Attack Detection (PASAD) anomaly detection method. Using a second-order linear time-invariant with time delay model, representative of water treatment dynamics, we design and simulate stealthy attacks. Our results highlight how factors such as system noise and inaccuracies in the attacker’s plant model influence the attack's stealthiness, underscoring the need for more robust detection strategies in industrial control environments.
\end{abstract}

\maketitle
\thispagestyle{empty}
\pagestyle{empty}


\section{Introduction}

Cyber-physical systems (CPS), such as those found in industrial control and water treatment plants, are increasingly exposed to sophisticated adversaries capable of launching stealthy cyberattacks. Among these, covert man-in-the-middle (MitM) attacks pose a particular challenge: by manipulating both control commands and sensor measurements, attackers can alter system behavior without triggering alarms. The effectiveness of such attacks is amplified when the adversary first engages in system identification (SI) to estimate the plant and controller dynamics, and then uses that information to minimize the perception of the effect caused by the attack \cite{alan}.\footnote{This is an extended version of~\cite{mattos2025design}. © ACM 2025. This is the author's version of the work. It is posted here for your personal use. Not for redistribution. The definitive Version of Record was published in the Proceedings of the 2025 Workshop on Re-design Industrial Control Systems with Security (ACM CCS Workshop), https://doi.org/10.1145/3733823.3764516.}

This paper investigates the design of covert MitM attacks that leverage SI techniques to learn the system’s input-output behavior. Once the attacker has an approximate model, they can predict the system’s response to malicious inputs and inject compensatory signals to cancel their observable effects, effectively concealing the attack~\cite{smith}. We assume that the SI phase has already been completed and focus on the second stage: designing attack strategies and using detection mechanisms to evaluate the success of the attack. In particular, we explore the limits of stealth under different model accuracies and detection thresholds and also assess the performance of advanced anomaly detectors in identifying these hidden manipulations. Note that considering the attack strategy as well as the detection method together allows the evaluation of the attack strategy and, dually, that of the detection method.


Recent incidents highlight the risks of set point manipulation in water infrastructure. In 2025, a reported incident in Norway involved attackers manipulating the reference values of a dam’s control system. At one point, the adversaries attempted to inject an implausible 999\% input, which made the attack easily detectable rather than covert.\footnote{\url{https://www.csoonline.com/article/4042449/russia-linked-european-attacks-renew-concerns-over-water-cybersecurity.html}} In November 2023, the Iranian-linked group Cyber Av3ngers compromised the Municipal Water Authority of Aliquippa, targeting a Programmable Logic Controller (PLC) that regulated water pressure and forcing operators into manual control.\footnote{\url{https://www.cisa.gov/news-events/alerts/2023/12/01/iranian-cyber-actors-cyberav3ngers-compromise-us-water-utility}}



To assess the resilience of water treatment plants against such threats, this work evaluates both the effectiveness of these covert attack strategies and the ability of detection mechanisms to identify anomalies. The proposed approach builds upon existing methods for covert controller design and attack detection~\cite{smith, pasad,kravchik2021efficient}, providing a structured approach for analyzing vulnerabilities in networked control systems (NCS).

This paper explores change-of-reference attacks \cite{alan} — demonstrating how an adversary can execute them while remaining undetected. The findings contribute to a broader understanding of cybersecurity risks in industrial control systems and highlight the need for stronger defenses against more sophisticated cyber threats.

\textbf{Related Work. } 
While a large body of research has focused on cyberattacks targeting the electrical grid~\cite{teixeira2015secure,liang2017review}, comparatively fewer studies address water treatment systems~\cite{tuptuk2021systematic}. These facilities often operate in a decentralized manner, with limited staffing, budget constraints, and outdated infrastructure~\cite{adepu2020investigation}. As a result, they lack comprehensive cybersecurity measures and remain vulnerable to adversaries capable of launching stealthy attacks. This motivates the need for systematic models and detection strategies tailored to the unique operational and economic characteristics of water treatment infrastructure.

The Secure Water Treatment (SWaT)~\cite{mathur2016swat} and Water Distribution (WADI)~\cite{li2019mad} datasets are among the most widely used benchmarks for studying cybersecurity in water infrastructure. These datasets provide detailed network traffic and physical sensor data collected from operational water treatment testbeds, and have been instrumental in advancing data-driven attack and detection techniques.

Nonetheless, as with any modeling effort, a trade-off must be made between expressivity and simplicity. In this work, we prioritize simplicity to facilitate a principled evaluation of how factors such as model estimation error and measurement noise affect both attack effectiveness and detection sensitivity.   Therefore, complementary to the SWaT and WADI  efforts, we study a less-explored but still realistic plant model introduced in~\cite{water_plant}. This model offers a simpler analytical structure that enables clearer insights into the dynamics of attack and defense. We apply the covert control scheme from~\cite{smith} in combination with the PASAD anomaly detection strategy from~\cite{pasad}, allowing for a control-theoretic analysis that highlights the interplay between system identification, attacker design, and detection performance.

\subsection*{Contributions}  

This paper makes the following three contributions to the study of cyber-physical security in water treatment systems:

\begin{enumerate}
    \item \textbf{Attack-Defense Scheme for Covert Control.}  
    We develop a structured approach that integrates the covert control architecture of Smith \cite{smith} for man-in-the-middle (MitM) attacks with a detection strategy based on the PASAD anomaly detector.
   Then, we evaluate how the accuracy of a system identification attack influences the interaction between covert MitM attacks and a detection system.
    %

    \item \textbf{Application    to Realistic Water Treatment Models.}  
    We instantiate the proposed approach using a linear time-invariant (LTI) model derived in previous work from real-world data collected at an operational water treatment facility. We design and evaluate  change-of-reference (CR)  stealthy attacks~\cite{alan}  and show how attackers can exploit system identification to remain undetected. To the best of our knowledge, this is the first work to design \textit{covert} control-based attacks tailored to water treatment infrastructure.

    \item \textbf{Detection Analysis with PASAD vs. CUSUM.}  
    We analyze the detection performance of PASAD under different attack profiles and compare it with the widely used CUSUM detector in process control applications~\cite{Montgomery_2020}. Our results demonstrate that PASAD is more effective at identifying gradual reference manipulation anomalies, particularly when the attacker leverages high-fidelity system identification. This comparison offers actionable insights into the choice of anomaly detection tools for cyber-physical systems.
\end{enumerate}

\textbf{Outline.} The remainder of the paper is organized as follows. Section~\ref{sec:wild} surveys real-world incidents in water and wastewater facilities. Section~\ref{sec:plantmodel} describes the dynamic model of the water treatment process used in this study.  
Section~\ref{sec:archi} presents the architecture of the covert man-in-the-middle attack, along with a control-theoretic formulation. Section~\ref{sec:eval} provides a detailed evaluation of attack and detection mechanisms, comparing PASAD and CUSUM under varying attack intensities and model mismatches. Finally, Section~\ref{sec:discussion} discusses assumptions and threats to validity and Section~\ref{sec:concl} concludes.

\begin{table}[t]
\centering
\caption{Summary of Shodan-related references and use cases for water infrastructure security analysis.}
\begin{tabular}{|p{0.25\textwidth}|p{0.3\textwidth}|}
\hline
\textbf{Description} & \textbf{Keywords (clickable link)} \\
\hline
General access to the Shodan platform & \href{https://www.shodan.io/dashboard}{https://www.shodan.io} \\
\hline
Search for ICS-tagged water systems   & \href{https://www.shodan.io/search?query=tag%3Aics+%22water%22+-I20100}{tag:ics ``water" -I20100} \\
\hline
 VNC  open in water-related devices & \href{https://www.shodan.io/search?query=%22water%22+-I20100+port%3A5900}{port:5900 ``water" -I20100} \\
\hline
  Telnet  open in water-related devices & \href{https://www.shodan.io/search?query=%22water%22+-I20100+port%3A23}{port:23 ``water" -I20100} \\
\hline
\end{tabular}
\label{tab:shodan_refs}
\end{table}

\section{Incidents in the Wild} \label{sec:wild}

In this section, we relate real incidents to the models considered in this work. 

\subsection{Data Collection and Curation} 
We collected data of cyberattack incidents over the past 11 years in water and wastewater facilities, leveraging artifacts such as news articles, reports, and technical papers~\cite{sikder2023deep,adepu2020investigation}. A total of 30 incidents were identified. Despite the small number of public incidents, once an incident becomes public it attracts significant attention from the newsmedia: we found an average of 52 artifacts per incident, using the DuckDuckGo search engine.\footnote{https://duckduckgo.com/} Those artifacts were semi-automatically processed with the help of Gemini v2.0 language model.



To better understand the exposure of water infrastructure to cyber threats, we conducted a reconnaissance study using Shodan~\cite{shodan} -- a search engine for internet-connected devices -- and datasets from the U.S. Environmental Protection Agency’s Facility Registry Service~\cite{epa}. By querying terms such as ``water'' and ``wastewater'', we identified approximately 8{,}000 potential water and wastewater treatment facilities across the United States. Among these, 26 facilities were explicitly labeled as Industrial Control Systems (ICS) endpoints by Shodan (Table~\ref{tab:shodan_refs}).



To further probe their attack surfaces, we searched for exposure of specific ports known to be insecure in ICS contexts. For instance, we identified at least one facility exposing port 502/tcp, used by the Modbus protocol—a communication standard widely deployed in SCADA systems that lacks encryption and authentication, making it highly susceptible to spoofing~\cite{502_is_insecure}. 
Additionally, we found 21 facilities exposing open Telnet ports, a deprecated protocol with well-known security flaws, and 47 cases where port 5900 (associated with VNC-based remote desktop access) was openly accessible (see Table~\ref{tab:shodan_refs}). These findings reinforce the notion that even basic security hygiene is lacking in many water infrastructure deployments, increasing the feasibility of adversaries gaining reconnaissance-level access and performing system identification attacks.

As another example of how open-source intelligence (OSINT) can be leveraged to assess cyber-physical risks, publicly available satellite imagery (e.g., from platforms such as Google Maps) can reveal the physical layout and geographic coordinates of water treatment facilities. When combined with network reconnaissance tools such as Shodan, this geolocation data can be correlated with Internet-facing services apparently operated by the same facility. Typical findings may include industrial control components (e.g., BACnet/IP routers), network gateways, and embedded web servers. The exposure of such services often reveals security weaknesses, such as outdated or weak encryption protocols (e.g., \texttt{TLSV1\_ALERT\_PROTOCOL\_VERSION}, self-signed SSL certificates), unprotected management interfaces (e.g., HTTP on port 8080), or ICS-specific vulnerabilities (e.g., unprotected BACnet services). Taken together, this combined OSINT approach—linking visual reconnaissance with network exposure analysis—enables precise system fingerprinting and highlights the potential for serious security breaches in critical infrastructure.



\subsection{Change-of-Reference (CR)}

 In this attack, the adversary alters the setpoint to manipulate the behavior of the system. By making gradual or minor changes, the adversary can, for example, degrade performance, increase costs, or create unsafe conditions.  

Although water treatment plants are increasingly targeted by cyberattacks, none of the 30 documented incidents analyzed in this study have led to actual contamination of the water supply. This is largely due to the ability of operators to switch to manual control in response to abnormal activity. For example, in the 2023 attack on the Municipal Water Authority of Aliquippa (Pennsylvania), the adversary successfully disabled a PLC  responsible for regulating water pressure, forcing a manual override to maintain operations. Similarly, during the 2024 incident at the Arkansas City Water Treatment Facility (Kansas), manual interventions prevented the compromise of chemical dosing processes~\cite{manual_switch_example}. These incidents, including the Florida incident discussed in the introduction, demonstrate that  attackers were able to perturb critical components inside the plants. This progression highlights an urgent need for proactive defense mechanisms, as future attacks may succeed if detection or response is delayed.


 \subsection{System Identification (SI)}


A key assumption in our work is that attackers are able to perform approximate system identification (SI) of the target infrastructure prior to launching covert attacks. By collecting and analyzing the plant's response to various control inputs, the adversary can derive an approximate model of the system’s behavior. 
 While incident reports and public sources rarely reference SI explicitly—such as defined in engineering contexts (e.g.,~\cite{matlab_sys_iden_def})—indirect evidence suggests that adversaries may indeed engage in such reconnaissance activities.

First, in several documented incidents, there is a noticeable delay between the initial system breach and the execution of the attack. For example, in the South Staffordshire Water incident, the attackers remained within the network environment for an extended period before initiating any observable disruption~\cite{lengthy_attack1,lengthy_attack2}. This dwell time is consistent with efforts to understand system behavior prior to action.

Second, some attackers have demonstrated the capability to manipulate system parameters directly, e.g., for SI purposes. In a recent campaign targeting the Texas cities of Hale Center, Lockney, and Abernathy, adversaries reportedly accessed and altered operational settings in the water control systems~\cite{direct_SI_manipulation1,direct_SI_manipulation2}. While these incidents do not constitute definitive proof of systematic SI, they suggest a growing sophistication in attacker behavior.

These findings support two plausible interpretations: either system identification is not yet a common tactic in attacks on water and wastewater (WWS) facilities, or, more likely, current forensic and monitoring techniques are insufficient to detect such preparatory activity. In either case, the increasing complexity of observed intrusions underscores the urgency of studying SI-based attacks and developing effective countermeasures.


\section{The Plant Model} \label{sec:plantmodel}

To illustrate the design and detection of covert attacks in water treatment plants, we consider the fluoridation process of a real-world drinking water treatment plant located in Peninsular Malaysia~\cite{water_plant}. Fluoride is added to treated water in carefully regulated doses, and maintaining its concentration within narrow limits is critical for public health and regulatory compliance. If these limits are violated due to a cyberattack there are well known undesirable consequences~\cite{cury_systemic_2019,iheozor-ejiofor_water_2024}. The plant data, collected over regular 15-minute intervals, was used to derive a dynamic model for the fluoride concentration in the treated water.

A second-order plus time delay (SOPTD) model was found to offer a good trade-off between prediction accuracy and analytical simplicity. The identified linear time-invariant (LTI) transfer function is given by:
\begin{equation}
P = \frac{0.93}{(1.07s + 1)(0.34s + 1)} e^{-0.45s}, \label{eq:soptd}
\end{equation}
where \( P \) represents the relationship between the fluoride dosing input and the resulting fluoride concentration at the plant outlet. The model reflects key characteristics of the physical process, including a time delay of 0.45 hours (27 minutes), a steady-state gain of 0.93 (indicating partial absorption or precipitation of fluoride), and two dominant time constants.

Using this plant model, a proportional-integral-derivative (PID) controller was designed in \cite{water_plant} via the Internal Model Control (IMC) methodology:
\begin{equation}
G_C(s) = 1.69 \left( 1 + \frac{1}{1.41s} + 0.26s \right).
\end{equation}
To compensate for the process dead time, a Smith Predictor structure was implemented. While the full plant dynamics are modeled by the SOPTD transfer function \( P \), the internal model used within the Smith Predictor was a simplified first-order plus time delay (FOPTD) approximation:
\begin{equation}
    \bar{G}(s) = \frac{0.62}{s + 0.64} e^{-0.29s}. \label{eq:approx}
\end{equation}
Approximation~\eqref{eq:approx} retains the dominant dynamics and delay characteristics of the original model, allowing for effective dead-time compensation while maintaining analytical and computational tractability within the controller design. It is important to note, however, that the original second-order model~\eqref{eq:soptd} is retained for evaluating the system’s closed-loop behavior and simulating attack and detection scenarios \cite{water_plant}.

It is worth emphasizing that, although the plant model is derived from a water treatment process, the methodology and analysis presented in this work are not restricted to this specific domain. Since the system is modeled as a LTI process, the proposed study and results are generalizable to LTI systems in general.

\begin{table}[t]
\centering
\caption{Notation used in MitM attack architecture}
\begin{tabular}{ll}
\hline
Symbol & Description \\
\hline
\( P \), \( C \) & Nominal plant and controller (LTI systems) \\
\( y_{\textrm{ref}} \) & Reference signal \\
\( y_{\textrm{m}} \) & True plant output, $y_{\textrm{m}}=y + n$  \\
\( y_{\textrm{ma}} \) & Manipulated plant output, $y_{\textrm{ma}}=y_{\textrm{m}} - \gamma$  \\
\( u_{\textrm{c}} \), \( \mu \) & Controller and covert controller outputs  \\
\( w \), \( n \) & External disturbance and measurement noise \\
\( \Pi_u \) & Attacker's identified model of the plant \\
{\( \alpha \)} & Model scaling factor (multiplicative model error \\
& due to imprecise system identification by the attacker) \\
\( \Theta \) & Covert controller designed by the attacker \\
\( \gamma_{\textrm{ref}} \) & Attacker’s reference deviation signal \\
\( \gamma \) & Predicted response of the plant to injected signal, \\
& $\gamma = \Pi_u \cdot \mu$ \\
\( \mu \) & Injected actuation computed by the attacker \\
\hline
\end{tabular}
\label{tab:notation_mitm}
\end{table}


\section{Architecture of Man-in-the-Middle Attack} \label{sec:archi}

We adopt the attack architecture proposed in~\cite{smith}, in which an adversary acts as a covert agent inserted between the plant and the nominal controller in a networked control system. In such systems, both the actuation commands issued by the controller and the sensor measurements sent from the plant are transmitted over a communication network. This exposes them to potential interception and manipulation by a man-in-the-middle (MitM) attacker. We briefly recapitulate the approach proposed in \cite{smith}, to which the reader is referred for further details.

The plant and the nominal controller are both modeled as linear time-invariant (LTI) systems (see Fig.\ref{fig:smith_blk_diag}). The reference signal is denoted by \( y_{\textrm{ref}} \), and the measured plant output, including sensor noise, is denoted by \( y_{\textrm{m}} \). The attacker manipulates the signal \( y_{\textrm{m}} \), with the objective of avoiding detection, to produce \( y_{\textrm{ma}} \), which is the signal observed by the controller. Similarly, the controller generates an actuation command \( u_{\textrm{c}} \), which the attacker modifies into the input \( u \) that is ultimately applied to the plant.

The covert agent performs a system identification procedure to estimate the plant dynamics, producing an approximate model \( \Pi_u \). It then uses this model to compute a compensating signal based on its desired attack objective. Let \( \Theta_{\textrm{ref}} \) denote the attacker's covert controller. The signal \( \gamma_{\textrm{ref}} \) encodes the attacker’s reference objective—i.e., the deviation they wish to impose—while \( \gamma \) represents the expected response of the estimated plant model to the injected signal. 
The signal \( \mu \) represents the attacker’s injected input, and \( \gamma \) is the attacker’s prediction of how the plant will respond to \( \mu \). By injecting \( \mu \) and subtracting \( \gamma \) from the measured plant output \( y_{\textrm{m}} \), the adversary ensures that \( y_{\textrm{ma}} \) appears consistent with the nominal closed-loop behavior, effectively concealing the presence of the attack.

\begin{figure}[t]
    \centering
    \includegraphics[width=0.5\linewidth]{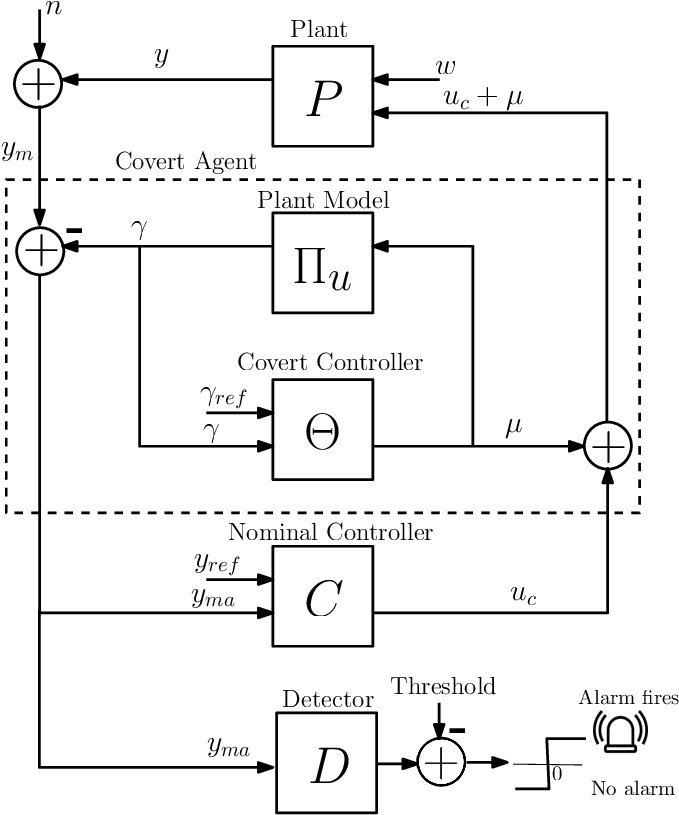}
        \caption{Architecture of the MitM attack, based on  \cite{smith}. The detector block $D$ is described in Section~\ref{sec:detection-design}.}
    \label{fig:smith_blk_diag}
\end{figure}

This architecture enables various attack strategies, such as the change-of-reference attack, in which the attacker maintains a constant additive reference \(|\gamma_{\textrm{ref}}| > 0\), steering the plant to an alternate reference value with an offset, $y_{\text{ref}} + \gamma_{\text{ref
}}$, from the original one. To remain undetected, the attacker can perform sensor spoofing, falsifying the measurement data so that the system appears to operate normally and no alarms are triggered.

It is worth noting that the possibility of compromising the confidentiality and integrity of data in protocols often used in this type of control (e.g., PROFINET) has already been demonstrated in the literature \cite{collantes2015protocols,Alsabbagh_2021}, indicating the feasibility of establishing a MitM shown in Figure \ref{fig:smith_blk_diag} from the network perspective.

\textbf{Attacker Model.} In our formulation, the covert agent is assumed to have the capability to both intercept and manipulate actuation commands as well as sensor measurements exchanged over the communication network. This corresponds to an adversary with read-and-inject access to the communication channels between the plant and the controller. In practice, such capabilities could be realized by augmenting actuators or modifying sensors.  Another possibility is to corrupt the programmable logic controllers (PLCs) used by the nominal controller \cite{smith}.

\subsection{Attack Design: Criteria for Determining Attack Parameters}

The likelihood of an attack being detected depends on four main factors that will be explored in this paper:

\begin{enumerate}
        
    \item \textbf{System Noise} – The inherent noise in the system can mask deviations caused by the attack, making it more difficult to distinguish between normal fluctuations and malicious interference.
    
    \item \textbf{Accuracy of the Attacker’s Identified Model} – A more precise model allows the attacker to better predict the system's response, reducing the chance of detection.

    \item \textbf{Intensity of the attack}
       A larger $\gamma_{ref}$ leads to more noticeable deviations from normal operation, making any model inaccuracies more apparent and increasing the risk of detection.
    \item \textbf{Dynamics of the Covert Controller} – In the case of a CR attack, the time taken for the controlled variable to reach the attacker's setpoint plays a crucial role. A slower change results in a smaller deviation at any given moment, leading to a lower detection statistic and making the attack harder to detect.
\end{enumerate}

While the first two factors are determined by the system itself, the latter two are within the attacker's control to some extent.

Given a specific attack objective, such as changing the reference by 10\%, an attacker can optimize their approach based on the detection threshold. By improving the accuracy of system identification ({\it i.e.}, $\Pi_u \rightarrow P$), the attack can be designed to remain undetected. 


\begin{figure*}[t]
    \centering
    \includegraphics[width=0.9\linewidth]{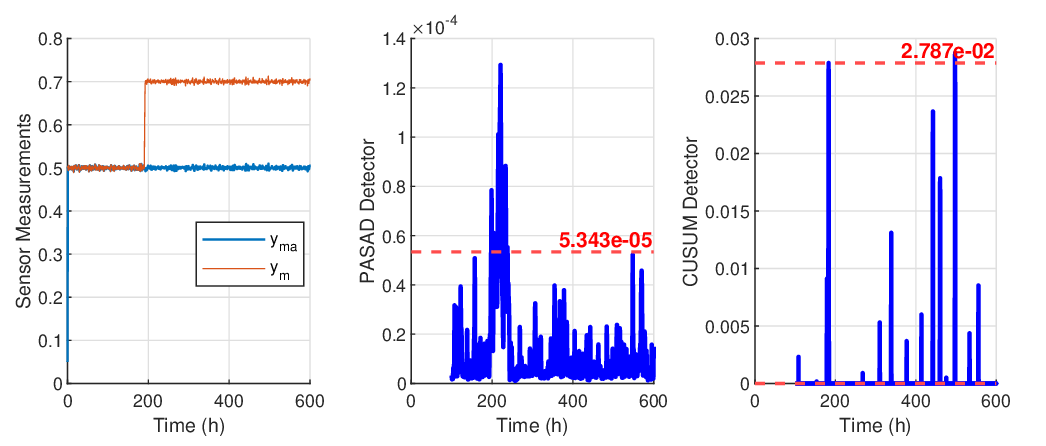} \vspace{-0.15in}
    \caption{From left to right: sensor measurements before ($y_m$) and after ($y_{ma}$) manipulation; PASAD detector; CUSUM detector, with their respective thresholds given by the red dashed line. In this example, the maximum PASAD  and CUSUM detection statistics are $1.29 \cdot 10^{-4}$ and  $2.87 \cdot 10^{-2}$, respectively.}\label{fig:pasad_cusum-thresholds} \vspace{-0.12in}
\end{figure*}
\subsection{Detection Design}\label{sec:detection-design}

Detection mechanisms exhibit differing levels of sensitivity to such covert attacks~\cite{pasqualetti2013attack}. The CUSUM detector, which relies on detecting abrupt shifts in the statistical properties of residuals, often fails to detect these attacks, as the adversary precisely cancels the visible effect of the injected signal. In contrast, PASAD~\cite{pasad} is designed to detect low-amplitude or slow-moving anomalies by analyzing lag correlations and frequency-domain features in the sensor time series. As we show in later sections, PASAD consistently outperforms CUSUM in detecting the covert strategies described above, particularly under realistic conditions with noise and model uncertainty. To our knowledge, PASAD has not been tested under such conditions.

\subsection{Model Error, Noise and Detectability}

We analyze how model mismatch and noise jointly affect detectability. To isolate the role of model error, Lemma~\ref{theo:1} considers an idealized setting without noise or disturbances. The impact of noise is then addressed in Theorem~\ref{theo:2} and further explored through simulations in the next section.

\begin{lemma}[Residual Leakage Due to Model Error] \label{theo:1}
Let the true plant be \( P \), and let the attacker use \( \Pi_u = \alpha \cdot P \), with \( \alpha > 0 \). If a covert input \( \mu \) is injected and the attacker sets \( y_{\mathrm{ma}} = y - \gamma \), where \( \gamma = \Pi_u \cdot \mu \), then under zero noise/disturbance, 
$
r = y_{\mathrm{ma}} - y_{\mathrm{nominal}} = (1 - \alpha) \cdot P \cdot \mu,$  
and
\begin{equation}
\|r\|_2 = |1 - \alpha| \cdot \|P \cdot \mu\|_2.
\end{equation}
\end{lemma}

\begin{proof}
The plant output is \( y = y_{\mathrm{nominal}} + P \cdot \mu \), and the attacker predicts \( \gamma = \alpha \cdot P \cdot \mu \). Thus,
\begin{equation}
y_{\mathrm{ma}} = y - \gamma = y_{\mathrm{nominal}} + (1 - \alpha) \cdot P \cdot \mu,
\end{equation}
yielding the stated residual.
\end{proof}

While the above lemma shows that model mismatch generates a detectable residual, it does not consider the impact of noise. However, noise introduces ambiguity—when residuals are small enough, they can be masked by the nominal noise envelope.

We now consider how noise affects detectability. 
To build intuition, consider a detector that raises an alarm if the residual signal exceeds a threshold \( \delta > 0 \), and let the observed residual be \( r + n \), where \( n \) is additive noise. By the triangle inequality,
\begin{equation}
    \|r + n\|_2 \leq \|r\|_2 + \|n\|_2.
\end{equation}
So, a sufficient condition for the attack to remain undetected is \begin{equation}
\|r\|_2 < \delta - \|n\|_2. \label{eq:condition}
\end{equation}
If the residual induced by the model mismatch is small enough, it can be hidden below the threshold by the noise. Although \( \|n\|_2 \) is a random quantity,   condition~\eqref{eq:condition} captures a basic deterministic mechanism that enables probabilistic stealth.  Next, we establish a probabilistic bound: 
\begin{theorem}[Probabilistic Stealth  Success Under Noise] \label{theo:2}
Let the residual under attack be \( r = (1 - \alpha) P \cdot \mu \), and suppose the additive noise \( n \sim \mathcal{N}(0, \sigma^2 I_d) \) is i.i.d. Gaussian of dimension \( d \). If the detector raises an alarm when \( \|r + n\|_2 > \delta \), then the probability that the attack remains undetected satisfies
\[
\mathbb{P}(\|r + n\|_2 \leq \delta) \geq 1 - \exp\left(-\frac{(\delta - \|r\|_2 - \sqrt{d} \sigma)^2}{2 \sigma^2}\right),
\]
for any \( \delta > \|r\|_2 + \sqrt{d} \sigma =\|(1-\alpha)P \cdot \mu\|_2 + \sqrt{d} \sigma  \).
\end{theorem}

\begin{proof}
Let \( X = \|r + n\|_2 \). Since \( n \sim \mathcal{N}(0, \sigma^2 I_d) \), the expected norm satisfies \( \mathbb{E}[\|n\|_2] = O(\sqrt{d} \sigma \)), and \( X \) concentrates around \( \|r\|_2 + \sqrt{d} \sigma \). Applying a sub-Gaussian tail bound yields
$\mathbb{P}(X > \delta) \leq \exp\left(-{(\delta - \|r\|_2 - \sqrt{d} \sigma)^2}/({2 \sigma^2})\right).$ 
Taking the complement gives the desired bound.
\end{proof}

In summary, detection is fundamentally governed by the trade-off between model mismatch and noise level~\cite{oliveira2020identification}. When residuals are small relative to the noise-induced threshold, the attack becomes increasingly difficult to detect.
%
%
%
In the next section, we explore this behavior empirically by varying both the model error and the noise standard deviation in a grid of simulations.

\section{Attack and Detection Evaluation} \label{sec:eval}

We now evaluate the effectiveness of the PASAD and CUSUM detectors against the covert attack architecture described in the previous section. The analysis is based on a comprehensive simulation study using the plant model presented in Section~III, and includes change-of-reference (CR) attack scenarios. Two key parameters are varied: the attack intensity, represented by \( \gamma_{\textrm{ref}} \), and the model estimation error, modeled as a multiplicative deviation from the true plant coefficients, implying that the model error occurs in the same frequency band used for the controller design. The delay is assumed to be known because it is easy to estimate.
 

We conduct 10{,}201 simulations across a grid of 101 values of \( \gamma_{\textrm{ref}} \in [-0.5, 0.5] \), corresponding to \(-100\%\) to \(100\%\) CR. Additionally, 101 multiplicative model errors $\alpha$ were considered within the range \([0.9, 1.1]\), representing up to \(\pm 10\%\) parametric uncertainty. The system operated with a sample time of 0.1, and the white noise power was fixed at \(10^{-9}\). To ensure reproducibility, all simulations used a fixed noise seed. For simplicity, the covert controller was set equal to the nominal controller. 
We assumed $w=0$, and the reference for \(C\) was set as \( y_{\textrm{ref}} = 0.5 \).  

Regarding the parameters of the PASAD algorithm, training was performed using the first 1{,}000 samples after reaching steady-state, with the attack starting at sample 2{,}000. The parameters were set as \( r = 26 \) and \( L = N/2 \).  
For the CUSUM detector, the reference value was defined as \( k = 0.3 \cdot \sigma_s \), where \( \sigma_s \) is the standard deviation of sensor measurements during the training phase (see Table~\ref{tab:simulpar}).

For each simulation, the detectors are evaluated on the same measurement stream. Detection thresholds for both PASAD and CUSUM are chosen slightly above the maximum statistic observed in attack-free runs, as shown in Figure \ref{fig:pasad_cusum-thresholds}, which guarantees a perfect true negative rate (TNR = 1) and zero false positives (FPR = 0).

\begin{table}[t]
\centering
\caption{Simulation and detector parameters.} \label{tab:simulpar}
\resizebox{0.5\textwidth}{!}{%
\begin{tabular}{ll}
\hline
\textbf{Parameter} & \textbf{Value / Description} \\
\hline
Simulations & \(10{,}201\) (\(101\times101\) grid) \\
Change-of-ref.\ \(\gamma_{\mathrm{ref}}\) & \([-0.5,\,0.5]\) (\(-100\%\) to \(100\%\)) \\
Model scaling error & \\
(multiplicative error) \(\alpha\) & \([0.9,\,1.1]\) (\(\pm10\%\) error) \\
Sample time & \(0.1\) \\
White noise power & \(10^{-9}\) \\
Noise seed & Fixed (reproducible) \\
Disturbance \(w\) & \(0\) (neglected) \\
Reference \(y_{\mathrm{ref}}\) & \(0.5\) \\
Covert controller & Same as nominal \\
\hline
\multicolumn{2}{c}{\textbf{PASAD parameters}} \\
\hline
Training samples & First \(1{,}000\) (after steady state) \\
Attack start & Sample \(2{,}000\) \\
Lag \(r\) & \(26\) \\
Embedding length \(L\) & \(N/2\) \\
\hline
\multicolumn{2}{c}{\textbf{CUSUM parameters}} \\
\hline
Reference value \(k\) & \(0.3\,\sigma_s\) (sensor std.\ during training) \\
\hline
\end{tabular}%
}
\label{tab:sim_params}
\end{table}

Figure~\ref{fig:pasad_cusum-thresholds} illustrates the impact of the attack on sensor readings (Fig.~\ref{fig:pasad_cusum-thresholds} (left)) and the corresponding response of two detection mechanisms. The first subplot shows the sensor measurements before manipulation ($y_{\textrm{m}}$) and after manipulation ($y_{\textrm{ma}}$), highlighting how the attacker’s compensation effectively conceals the deviation from the reference. The second and third subplots display the detection statistics  from PASAD and CUSUM, respectively, along with their thresholds (in red). Despite the presence of model mismatch and noise, the PASAD detector exhibits a clear excursion beyond the threshold shortly after the attack begins, indicating successful detection.  In contrast, CUSUM also detects the attack, but with a significantly longer delay and a much smaller margin — its maximum detection statistic remains much closer to the threshold compared to PASAD. This underscores PASAD's heightened sensitivity to subtle and gradual deviations, which traditional detectors like CUSUM often fail to detect promptly.

\subsection{Detection Performance under Model Mismatch}

In this scenario, the first set of experiments evaluates how detection performance degrades or improves depending on the combination of attack intensity and model mismatch. Figures 
\ref{fig:cusum-heatmap-error} (resp.~ ~\ref{fig:pasad-heatmap}) show the maximum CUSUM (PASAD) detection statistic across the parameter grid. As the model error increases, the attacker’s ability to compensate for the injected signal deteriorates, increasing the likelihood of detection.

\begin{figure}[htb]
    \centering
    \includegraphics[width=0.5\linewidth]{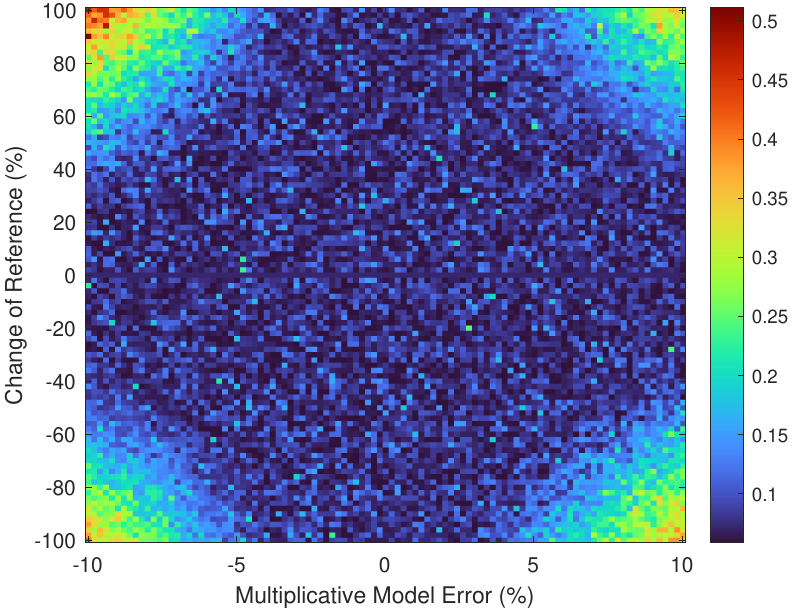}
    \caption{Maximum CUSUM detection statistic as a function of change-of-reference and multiplicative model error.}
    \label{fig:cusum-heatmap-error}
\end{figure}

\begin{figure}[htb]
    \centering
\includegraphics[width=0.5\linewidth]{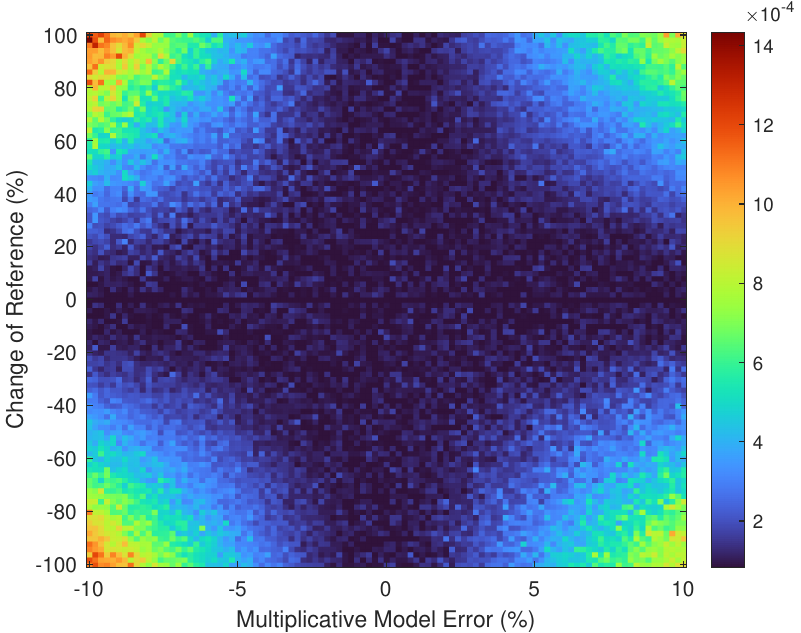}
    \caption{Maximum PASAD detection statistic as a function of change-of-reference and multiplicative model error.}
    \label{fig:pasad-heatmap} \vspace{-0.1in}
\end{figure}

 While CUSUM is more sensitive than PASAD to small model errors in some regions, it suffers from irregular behavior, particularly under noise. This results in the patchy, ``Swiss-cheese"-like appearance of the CUSUM heatmap (Fig.~\ref{fig:cusum-heatmap-error}), highlighting a lack of robustness.

\begin{figure}[htb]
    \centering    \includegraphics[width=0.5\linewidth]{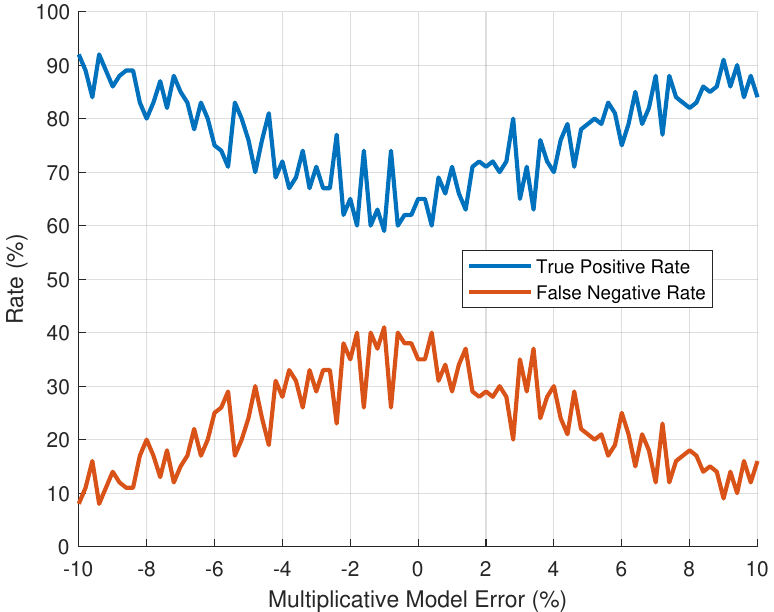}
    \caption{CUSUM classification results. Dual-threshold detection using both positive and negative excursions.}
    \label{fig:cusum-counts} \vspace{-0.1in}
\end{figure}

\begin{figure}[htb]
    \centering
\includegraphics[width=0.5\linewidth]{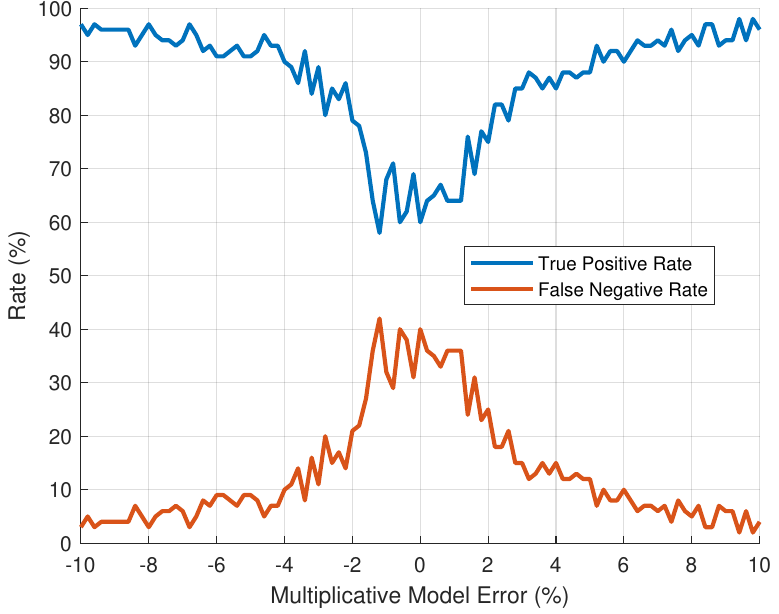}
    \caption{PASAD classification results: proportion of true positives (TP) and false negatives (FN) under varying attack intensities and model errors.} \vspace{-0.1in}
    \label{fig:pasad-counts}
\end{figure}

Figure \ref{fig:cusum-counts} (resp. ~\ref{fig:pasad-counts}) shows the classification outcomes (True Positives and False Negatives); once again CUSUM displays more irregular behavior under noise than PASAD. 

\subsection{Detection Performance under Noise}\label{sec:detection-noise}
We repeat the same simulation process while keeping all parameters unchanged. However, instead of varying the multiplicative model error across 101 values, we now vary the noise power over a logarithmic scale in the range \([10^{-7}, 10^{-11}]\), using 101 samples. The multiplicative model error is fixed at~5\%.  The results for PASAD and CUSUM are reported in Figures~\ref{fig:pasad-detection-noise} and~\ref{fig:cusum-detection-noise}, respectively. 

    \begin{figure}[ht]
    \centering
\includegraphics[width=0.5\linewidth]{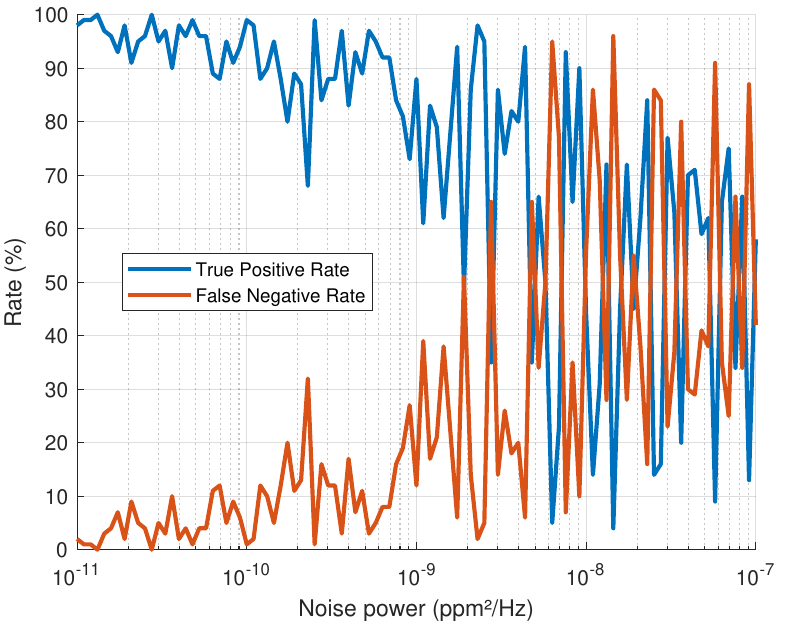}
    \caption{ PASAD classification results: proportion of true positives (TP) and
false negatives (FN) under varying attack intensities and noise power. The noise power axis represents the power spectral density, up to a constant factor.}
    \label{fig:pasad-detection-noise} \vspace{-0.1in}
\end{figure}

\begin{figure}[ht]
    \centering
    \includegraphics[width=0.5\linewidth]{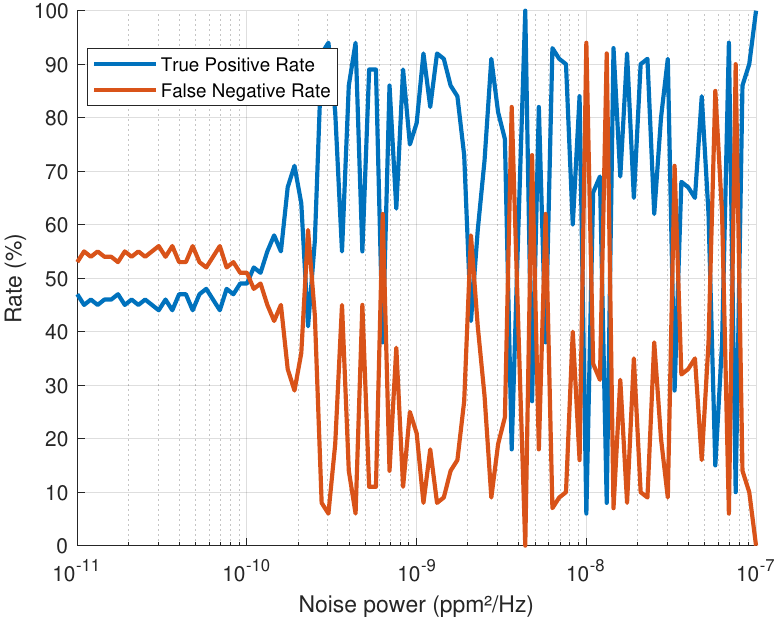}
    \caption{CUSUM classification results: proportion of true positives (TP) and
false negatives (FN) under varying attack intensities and noise power.  The noise power axis represents the power spectral density, up to a constant factor.}
    \label{fig:cusum-detection-noise}  \vspace{-0.1in} 
\end{figure}

Figures~\ref{fig:pasad-detection-noise} and~\ref{fig:cusum-detection-noise} confirm the observation from Section~\ref{sec:detection-noise} that PASAD is more robust to noise than CUSUM. Notably, CUSUM performs poorly at the start of Figure~\ref{fig:cusum-detection-noise}, as it requires recalibration of the \(k\) parameter for each noise level, whereas PASAD only needs threshold adjustment.

\subsection{Relation to Analytical Results}

The empirical findings in this section align with the theoretical predictions from Lemma~\ref{theo:1} and Theorem~\ref{theo:2}. Specifically, the observed scaling of the PASAD and CUSUM detection statistics  with respect to model mismatch is consistent with the residual norm bound $\|r\|_2 = |1 - \alpha| \cdot \|P \cdot \mu\|_2$ from Lemma~\ref{theo:1}. The approximate symmetry in the classification outcomes with respect to the model mismatch parameter $\alpha = 1$, as seen in Figures~\ref{fig:cusum-counts} and~\ref{fig:pasad-counts}, reflects the dependence of the residual on the absolute deviation $|1 - \alpha|$. Additionally, the heatmaps in Figures~\ref{fig:cusum-heatmap-error} and~\ref{fig:pasad-heatmap}, together with the noise-dependent curves in Figures~\ref{fig:pasad-detection-noise} and ~\ref{fig:cusum-detection-noise}, illustrate the exponential degradation in detectability as noise increases, in agreement with the probabilistic bound established in Theorem~\ref{theo:2}.

\subsection{PASAD vs CUSUM}

Overall, PASAD offers more consistent performance and broader coverage in detecting covert attacks, especially those characterized by slow or structured deviations. While CUSUM can outperform PASAD in specific regions of low model error and high signal-to-noise ratio, it is less reliable when the attack becomes more adaptive or the model uncertainty increases.

Notably, the PASAD detection landscape is approximately convex in the \( (\gamma_{\textrm{ref}}, 
 \text{model error}) \)-space, suggesting symmetrical growth in PASAD  detection statistics  with stronger attacks and greater mismatch. CUSUM, by contrast, shows less convex (i.e., flatter) behavior with respect to model error, thus being less sensitive to the magnitude of the injected attack.

\section{Discussion} \label{sec:discussion}

While our study presents a structured approach to evaluate covert MitM attacks and their detection in water treatment systems, several limitations must be acknowledged.

\textbf{Simplified system model and identification error. }
 We intentionally adopted a second-order plus time delay (SOPTD) model for the water treatment process to allow analytical tractability and controlled experimentation across noise and model mismatch parameters. However, this choice limits generality. Real-world water treatment systems may exhibit nonlinear, time-varying, and multi-variable dynamics, with more complex sources of noise and delay. Furthermore, the model mismatch was simplified as a scalar multiplicative deviation, which—while helpful for isolating analytical insights—does not reflect the richer structures typically encountered in system identification errors, such as unmodeled dynamics or structured parametric uncertainty. We are currently considering more expressive  system identification models and uncertainty representations to better reflect operational settings.

\textbf{Scope of attack and detection strategies. }
 Our approach assumes a single-stage covert change-of-reference (CR) attack, targeting a scalar process with a fixed control architecture. More sophisticated attack vectors—such as replay attacks, dynamic reference modulation, or state estimation poisoning—were not explored here. In addition,  although PASAD was shown to outperform CUSUM in our scenarios, the theoretical results in Lemma~\ref{theo:1} and Theorem~\ref{theo:2} focus on residual behavior under model error and noise but do not explicitly explain PASAD’s advantage over CUSUM. A formal link between detector characteristics and the residual's frequency-domain or correlation properties remains an open question.

\textbf{Application context. }  Further work is needed to extend and validate the approach across multi-stage treatment processes, varying control schemes, and additional  blueprints~\cite{siemens:pcs7_cyber_blueprint_water_2023}. Moreover, connecting attack dynamics to safety-critical outcomes, like regulatory threshold violations or human health impacts, would improve the specificity of the contribution.

\vspace{-0.1in}

 \section{Conclusion} \label{sec:concl}

This paper introduced a structured approach for designing and detecting covert MitM  attacks in water treatment plants, focusing on attacks guided by system identification. Using a simplified yet realistic model of a fluoridation process and the covert control architecture of~\cite{smith}, we showed how an adversary can stealthily manipulate control and sensor signals to execute undetected change-of-reference attacks. To counter such threats, we evaluated the PASAD anomaly detector against the traditional CUSUM approach. Simulations revealed that PASAD provides more robust detection, particularly under model mismatch and sensor noise, whereas CUSUM is more sensitive to parameter tuning and less effective against gradual or structured attacks.



We surveyed real incidents and public vulnerabilities, finding that water treatment systems remain  highly exposed to cyberattacks. While manual intervention has prevented contamination so far, attackers can already disrupt internal operations, and future attacks may succeed   if detection and response mechanisms are not improved.


As future work, we plan to extend this approach to accommodate more complex multi-stage processes and data-driven attack synthesis using benchmarks like the SWaT dataset~\cite{mathur2016swat}. Furthermore, formally exploring the scaling laws of  covertness under Neyman-Pearson detectors~\cite{ramtin2021fundamental}, investigating more sophisticated machine learning-based detectors  and exploring game-theoretic formulations for adaptive attacker-detector dynamics also remain promising directions for future work.
We are currently   implementing a proof-of-concept inspired by this paper  on a real PROFINET testbed, incorporating realistic network considerations to help bridge the gap between theory and practice.

 \textbf{Reproducibility and Artifact Availability. }  
 To foster  reproducibility, we   provide  all Matlab code, simulation scripts, and figure-generation tools used in this study.\footnote{\href{https://github.com/vctrmattos/mitm_water_codes}{\texttt{https://github.com/vctrmattos/mitm\_water\_codes}}}

\bibliographystyle{ACM-Reference-Format}

\bibliography{bibliogra}
\end{document}